\documentclass[conference]{IEEEtran}
\usepackage{cite}
\usepackage{orcidlink}
\usepackage{enumitem}
\usepackage{amsmath,amssymb,amsfonts}
\usepackage{amsthm}
\usepackage{mathtools}
\usepackage[bb=dsserif]{mathalpha}
\usepackage{bm}
\usepackage{mleftright}
\usepackage{siunitx}
\usepackage[only,llbracket,rrbracket]{stmaryrd}
\usepackage{xspace}
\usepackage{algorithmic}
\usepackage{graphicx}
\usepackage{textcomp}
\usepackage{xcolor}

\usepackage{pgfplots}
\usepackage{tikz}
\pgfplotsset{compat=1.17}
\usetikzlibrary{external,pgfplots.units,matrix,chains,positioning,intersections,calc,fit,dsp}
\pgfplotsset{table/search path={results/},}
\definecolor{wongblue}{RGB}{0, 114, 178}
\definecolor{wongorange}{RGB}{230, 159, 0}
\definecolor{wonggreen}{RGB}{0, 158, 115}
\definecolor{wongpurple}{RGB}{204, 121, 167}
\definecolor{wonglightblue}{RGB}{86, 180, 233}
\definecolor{wongvermillion}{RGB}{213, 94, 0}
\definecolor{wongyellow}{RGB}{240, 228, 66}

\definecolor{matlabblue}{rgb}{     0,     0.447, 0.741}
\definecolor{matlaborange}{rgb}{   0.85,  0.325, 0.098}
\definecolor{matlabyellow}{rgb}{   0.929, 0.694, 0.125}
\definecolor{matlabpurple}{rgb}{   0.494, 0.184, 0.556}
\definecolor{matlabgreen}{rgb}{    0.466, 0.674, 0.188}
\definecolor{matlablightblue}{rgb}{0.301, 0.745, 0.933}
\definecolor{matlabred}{rgb}{      0.635, 0.078, 0.184}

\dspdeclareoperator{dspshapexor}{
  \pgfmoveto{\pgfpointadd{\centerpoint}{\pgfpoint{0pt}{-\radius}}}
  \pgflineto{\pgfpointadd{\centerpoint}{\pgfpoint{0pt}{ \radius}}}

  \pgfmoveto{\pgfpointadd{\centerpoint}{\pgfpoint{-\radius}{0pt}}}
  \pgflineto{\pgfpointadd{\centerpoint}{\pgfpoint{ \radius}{0pt}}}

  \pgfusepathqstroke
}

\tikzset{%
  block/.style     = {draw,rectangle,align=center,inner sep=2mm},
  bigblock/.style  = {draw,rectangle,align=center,inner sep=2mm,minimum height=2.5em},
  hiergroup/.style = {draw,line width=0.3pt,inner sep=5mm,rectangle,rounded corners},
  dspxor/.style    = {shape=dspshapexor,line cap=rect,line join=rect,line width=\dspblocklinewidth,minimum size=\dspoperatordiameter},
}

\usepackage{glossaries}

\theoremstyle{plain}
\newtheorem{thm}{Theorem}
\newtheorem{lem}{Lemma}

\theoremstyle{definition}

\theoremstyle{remark}

\renewcommand\Pr{\operatorname{\mathcal P}}

\renewcommand*{\vec}[1]{\bm{#1}}
\newcommand*{\matr}[1]{\bm{#1}}
\newcommand{\bmat}[1]{\ensuremath{\begin{bmatrix}#1\end{bmatrix}}}

\DeclareMathOperator*{\argmax}{arg\,max}

\DeclareMathOperator{\E}{\mathbb E}
\DeclareMathOperator{\He}{\mathbb H}
\DeclareMathOperator{\MI}{\mathbb I}

\DeclarePairedDelimiter\abs{\lvert}{\rvert}
\DeclarePairedDelimiter\idxset{\llbracket}{\rrbracket}

\DeclarePairedDelimiter\set{\{}{\}}

\DeclarePairedDelimiterX{\infdivx}[2]{(}{)}{#1\;\delimsize\|\;#2}

\DeclareSIUnit{\belc}{Bc}
\DeclareSIUnit{\belm}{Bm}
\DeclareSIUnit{\bit}{bit}
\DeclareSIUnit{\sample}{S}
\DeclareSIUnit{\bpcu}{bpcu}

\newacronym{ask}{ASK}{amplitude-shift keying}
\newacronym{awgn}{AWGN}{additive white Gaussian noise}
\newacronym{ber}{BER}{bit error rate}
\newacronym{bicm}{BICM}{bit-interleaved coded modulation}
\newacronym{bidmc}{biDMC}{binary-input discrete memoryless channel}
\newacronym{bpcu}{bpcu}{bits per channel use}
\newacronym{bpsk}{BPSK}{binary phase-shift keying}
\newacronym{ccdm}{CCDM}{constant composition distribution matching}
\newacronym{crc}{CRC}{cyclic redundancy check}
\newacronym{dimc}{diMC}{discrete-input memoryless channel}
\newacronym{dmc}{DMC}{discrete memoryless channel}
\newacronym{dm}{DM}{distribution matching}
\newacronym{dpc}{DPC}{dirty paper coding}
\newacronym{dsnr}{dSNR}{design signal-to-noise ratio}
\newacronym{fec}{FEC}{forward error control}
\newacronym{fer}{FER}{frame error rate}
\newacronym{hy}{HY}{Honda-Yamamoto}
\newacronym{iff}{iff}{if and only if}
\newacronym{iid}{i.i.d.}{independent and identically distributed}
\newacronym{im}{IM}{intensity modulation}
\newacronym{ldpc}{LDPC}{low-density parity-check}
\newacronym{llps}{LLPS}{linear layered probabilistic shaping}
\newacronym{llr}{LLR}{log-likelihood ratio}
\newacronym{mac}{MAC}{multiple access channel}
\newacronym{mc}{MC}{Monte Carlo}
\newacronym{mi}{MI}{mutual information}
\newacronym{mlc}{MLC}{multilevel coding}
\newacronym{mlhy}{MLHY}{multilevel Honda-Yamamoto}
\newacronym{mlpc}{MLPC}{multilevel polar coding}
\newacronym{mmse}{MMSE}{minimum mean square error}
\newacronym{msd}{MSD}{multistage decoding}
\newacronym{ook}{OOK}{on-off keying}
\newacronym{pam}{PAM}{pulse-amplitude modulation}
\newacronym{pas}{PAS}{probabilistic amplitude shaping}
\newacronym{ps}{PS}{probabilistic shaping}
\newacronym{pcpas}{PC-PAS}{polar-coded probabilistic amplitude shaping}
\newacronym{qam}{QAM}{quadrature-amplitude modulation}
\newacronym{rcub}{RCUB}{random coding union bound}
\newacronym{scl}{SCL}{successive cancellation list}
\newacronym{sc}{SC}{successive cancellation}
\newacronym{se}{SE}{spectral efficiency}
\newacronym{sir}{SIR}{signal-to-interference ratio}
\newacronym{smi}{SMI}{symmetric mutual information}
\newacronym{snr}{SNR}{signal-to-noise ratio}
\newacronym{wlog}{w.l.o.g.}{without loss of generality}

\newcommand{\ie}{i.e.}

\newcommand{\eg}{e.g.}

\newcommand*{\Huuly}{\He(U^\ell_i|\vec V^\ell_i, \vec Y)}
\newcommand*{\Huul}{\He(U^\ell_i|\vec V^\ell_i)}

\newcommand*{\bbN}{{\idxset{N}}}
\newcommand*{\bbm}{{\idxset{m}}}
\newcommand*{\bbi}{{\idxset{i-1}}}
\newcommand*{\bbl}{{\idxset{\ell-1}}}
\newcommand*{\nbp}{{2^{-N^{\beta'}}}}
    
\begin{document}

\title{Multilevel Binary Polar-Coded Modulation Achieving the Capacity of Asymmetric Channels
}

\tikzexternaldisable
\author{\IEEEauthorblockN{Constantin Runge\,\orcidlink{0000-0001-8324-3945},
Thomas Wiegart\,\orcidlink{0000-0002-8498-6035},
Diego Lentner\,\orcidlink{0000-0001-6551-8925},
Tobias Prinz\,\orcidlink{0000-0002-9216-8075}}
\IEEEauthorblockA{Institute for Communications Engineering,
Technical University of Munich,
80333 Munich, Germany \\
\{constantin.runge, thomas.wiegart, diego.lentner, tobias.prinz\}@tum.de}
}

\maketitle
\tikzexternalenable

\begin{abstract}
A multilevel coded modulation scheme is studied that uses solely binary polar codes and Honda-Yamamoto probabilistic shaping. The scheme is shown to achieve the capacity of discrete memoryless channels with input alphabets of cardinality a power of two. The performance of finite-length implementations is compared to polar-coded probabilistic amplitude shaping and constant composition distribution matching.
\end{abstract}


\section{Introduction}

Reliable and power-efficient communication usually requires  \gls{ps} and/or geometric shaping .%
There are several ways to implement \gls{ps}, e.g., many-to-one mappings \cite[Sec. 6.2]{Gallager68}, trellis shaping \cite{forney92}, and others, see \cite[Sec.~II]{Bocherer15},
\cite{Gultekin20}. More recent schemes are \gls{pas} \cite{Bocherer15},
and \gls{hy} \gls{ps} \cite{Honda13} based on polar codes \cite{Stolte02, Arikan09}.

\Gls{pas} received significant attention from the optical fiber communications community and industry due to its performance and flexibility~\cite{Buchali:16,Boecherer:19}. \Gls{pas} requires a target distribution $P_X$ that factors as $P_X = P_A \cdot P_S$ so that $P_S$ is a uniform binary distribution. Usually ``$A$'' and ``$S$'' refer to the amplitude and sign of $X$, respectively, but more general choices are permitted. We focus on $P_S(-1)=P_S(1)=1/2$. An important component of \gls{pas} is a \emph{\gls{dm}} device that maps uniformly distributed bits to real-alphabet symbols with distribution $P_A$, e.g., a \gls{ccdm} device \cite{Schulte16}. These symbols are then protected with the parity bits of a systematic \gls{fec} code. Each parity bit is uniformly distributed and chooses one of two signs $S$ so that $X=A\cdot S$ has $P_X(x)=P_X(-x)$. \Gls{pas} in general does not allow for asymmetric $P_X$.

The \gls{hy} scheme generates asymmetric $P_X$ by performing joint \gls{dm} and \gls{fec}. The scheme achieves the capacity of general \glspl{bidmc} \cite{Honda13} and has excellent performance for short block lengths. For instance, see \cite{Wiegart19} that compares the performance of different schemes for \gls{ook} modulation over \gls{awgn} channels. 
An earlier scheme by Sutter et al. \cite{Sutter-ITW12} also achieves the capacity of \glspl{bidmc}. 
This scheme concatenates two separate polar codes for \gls{fec} and \gls{dm} which reduces the error exponent by a factor of two as compared to \gls{hy} coding \cite{Honda13}.

Polar codes can be extended to higher-order modulation by using \gls{mlc} \cite{Seidl13}.
In this paper, we study a \gls{mlhy} coding scheme which is amenable to practical implementation. %
Our contributions are two-fold.
First, we prove that \gls{mlhy} coding achieves the capacity of general \glspl{dmc} with $M=2^m$-ary channel inputs.
Second, we compare the \gls{dm} performance and shaping gains of \gls{pas} and \gls{mlhy} coding for short block lengths.
We evaluate the performance with unipolar ($X \geq 0$) and bipolar ($X\in\mathbb R$) modulation over \gls{awgn} channels. The proposed scheme performs on-par with polar-coded \gls{pas} \cite{Prinz17} and does not need a \gls{dm} device.

We remark that several polar coding architectures, including multilevel ones, were studied in \cite{Bocherer17,iscan18comml,iscan19access,icscan2020sign,Matsumine19,Bohnke20,Sener21} but these papers do not consider capacity proofs.
Using polar lattice codes, a capacity proof for channels whose input alphabets have a lattice structure is given in \cite{Liu19}.
Our proof and the proof in \cite{Liu19} are both based on the idea that each bitlevel polarizes, but we note that \gls{mlhy} coding is not restricted to lattice inputs. This makes our proof simpler and more general.
Note that the scheme of \cite{Liu19} is effectively a special case of the \gls{mlhy} coding studied here for the case of \gls{ask} modulation, Gaussian $P_X$ and a set-partitioning bit-mapping \cite{Ungerbock82}.

Instead of using \gls{mlhy} coding, the capacity of \glspl{dmc} can also be achieved by combining \gls{hy} coding with non-binary kernels \cite{Sasoglu09,Park13,Honda13}.
However, binary polar codes are preferred in practice because non-binary polar codes and decoders are complex to implement and design \cite{Gulcu18}, \cite{Yuan18}.

Polar-coded modulation has also been studied in the context of multiple-access channels \cite[Sec.~V]{Abbe12}
where each bitlevel of a channel input symbol corresponds to one user. 
Based on this approach, the authors of \cite{Abbe12} describe a \gls{mlc} scheme that achieves the symmetric capacity of \glspl{dmc} with $M=2^m$-ary channel inputs using independent binary polar codes for each bitlevel.
For transmission over the \gls{awgn} channel, they combine the scheme with a many-to-one mapping, which is not desirable in practice.

This paper is organized as follows.
Section~\ref{sec:pre} gives an overview of polar coding concepts.
Section~\ref{sec:dimc} shows that the \gls{mlhy} scheme achieves capacity and develops error exponents.
Finally, Section~\ref{sec:shaping} treats polar codes for short block lengths and provides numerical results.

\section{Preliminaries}
\label{sec:pre}
\subsection{Notation}

Random variables are written with upper case letters such as $X$. Their alphabet, distribution, and realizations are written as $\mathcal{X}$, $P_X$, and $x$, respectively. Vectors are denoted by bold symbols such as $\vec x$.
$\mathcal X^C$ and $\abs{\mathcal X}$ are the complement and cardinality of $\mathcal X$, respectively.
A set difference is denoted as $\mathcal X \setminus \mathcal Y = \mathcal X \cap \mathcal Y^C$.
An index set from $1$ to $N$ is denoted as $\bbN \triangleq \set{1,\dots,N}$.
A set $\mathcal S$ may select entries of a vector, creating a substring $\vec x_{\mathcal S}$ with length $\abs{\mathcal S}$, \eg, $\vec x_\bbN$. An event $\mathcal E$ has probability $\Pr(\mathcal E)$.

We denote by $\He(X)$, $\He(X|Y)$, and $\MI(X; Y)$ the entropy of $X$, the entropy of $X$ conditioned on $Y$ and the \gls{mi} of $X$ and $Y$, respectively.
The conditional Bhattacharyya parameter \cite{Arikan10} is defined as $Z(X|Y) = 2\E\mleft[\sqrt{P_{X|Y}(0|Y)P_{X|Y}(1|Y)}\mright]$ and satisfies \cite{Arikan10} \cite[Lem.~6]{Liu19}
\begin{gather}
    Z(X|Y)^2 \leq \He(X|Y) \leq Z(X|Y) \\
    Z(X|Y,S) \leq Z(X|Y) \label{eq:bhattacharyya_conditional}
\end{gather}
where $X,Y,S \sim P_{X,Y,S}$.

\subsection{Polarization and Polar Coding}

Polar codes \cite{Stolte02}, \cite{Arikan09} are linear block codes of length $N=2^n$ for $n\in\mathbb N$.
They are defined via the polar transform that maps a vector $\vec u\in\mathbb F_2^N$ to a codeword $\vec x\in\mathbb F_2^N$ 
\begin{equation}
    \vec x = \vec u \matr G_N \textnormal{ with } \matr G_N = \matr B_N \bmat{1 & 0\\ 1 & 1}^{\otimes n}
\end{equation}
where $\matr B_N$ is the bit-reversal matrix as in \cite{Arikan09}, and where $\matr F^{\otimes n}$ is the $n$-fold Kronecker product of $\matr F$.
The polar transform satisfies $\matr G_N^{-1} = \matr G_N$.
For encoding, we will consider the codeword $\vec x$ to have \gls{iid} entries.
The codeword $\vec x$ is transmitted over $N$ uses of a \gls{bidmc} $W\colon X \to Y$ resulting in a vector of channel observations $\vec y\in\mathcal Y^N$.
Consider the sets
\begin{align}
    \mathcal L_{U|Y} &= \set{i\in\bbN: Z(U_i|\vec U_\bbi, \vec Y) < \delta_N}     \textnormal{,} \\
    \mathcal H_{U|Y} &= \set{i\in\bbN: Z(U_i|\vec U_\bbi, \vec Y) > 1 - \delta_N}
\end{align}
with  $\delta_N \triangleq 2^{-N^\beta}$ for any $\beta<\frac12$.
It is known \cite[Eqs.~(38), (39)]{Honda13} that these index sets polarize, \ie, we have
\begin{align}
    \lim_{N\to\infty} \frac1N\abs*{\mathcal L_{U|Y}} &= 1 - \He(X|Y) \label{eq:asym_polarization1} \\
    \lim_{N\to\infty} \frac1N\abs*{\mathcal H_{U|Y}} &= \He(X|Y)     \textnormal{.} \label{eq:asym_polarization2}
\end{align}

The encoder places the data bits on the reliable bit positions of $\vec u$, \ie, those with $Z(U_i|\vec U_\bbi, \vec Y)\approx0$.
The remaining positions in $\vec u$ are frozen, \ie, set to fixed values.
The receiver uses \gls{sc} decoding of the non-frozen bits via $\hat u_i = \argmax_{u} P_{U_i|\vec U_\bbi,\vec Y}(u|\vec{\hat u}_\bbi, \vec y)$. %

Honda and Yamamoto \cite{Honda13} consider two more sets:
\begin{align}
    \mathcal L_{U} &= \set{i\in\bbN: Z(U_i|\vec U_\bbi) < \delta_N}    \\
    \mathcal H_{U} &= \set{i\in\bbN: Z(U_i|\vec U_\bbi) > 1 - \delta_N} . \label{eq:mathcal-HU}
\end{align}
Note that \eqref{eq:asym_polarization2} and \eqref{eq:mathcal-HU} yield $\abs*{\mathcal H_U}/N \overset{\footnotesize N\to\infty}{\to} \He(X)$.
With \eqref{eq:bhattacharyya_conditional}, we have $\mathcal L_U \subseteq \mathcal L_{U|Y}$ and thus $\mathcal L_U \cap \mathcal H_{U|Y} = \emptyset$.
For the ``data'' set  $\mathcal I = \mathcal H_U \cap \mathcal L_{U|Y}$, we thus have \cite[Thm.~1]{Honda13}
\begin{equation}
    \lim_{N\to\infty} \frac1N\abs*{\mathcal I} = \MI(X; Y).
\end{equation}
To achieve capacity, Honda and Yamamoto chose the $\mathcal I$ bits as data bits and the remaining bits in $\vec u$ randomly with probability $P_{U_i|\vec U_\bbi}(\cdot|\vec u_\bbi)$.
To calculate these probabilities, the same \gls{sc} structure as for decoding is employed. The random bits must be known to the receiver. We describe the method in more detail in Section~\ref{sec:enc_dec}.

\subsection{Conditional Polarization}

We next consider conditional polarization which helps to prove our main results.

\begin{lem}
    \label{lem:polarize_conditional_mi}
    Let $X\in\mathbb F_2$ and $S$ be the input to a \gls{bidmc} $W\colon X \to Y$ with side information $S$ with joint distribution $(X,Y,S) \sim P_{Y|X,S}P_{X|S}P_S$.
    Let $(\vec X, \vec S, \vec Y)$ be $N$ \gls{iid} realizations of $(X, S, Y)$ and let $\vec U = \vec X \matr G_N$.
    Then all the index sets
    \begin{align}
        \mathcal L_{U|S}   &= \set{i \in \bbN: Z(U_i|\vec U_\bbi, \vec S) < \delta_N} \label{eq:lus} \\
        \mathcal H_{U|S}   &= \set{i \in \bbN: Z(U_i|\vec U_\bbi, \vec S) > 1 - \delta_N}  \label{eq:hus} \\
        \mathcal L_{U|S,Y} &= \set{i \in \bbN: Z(U_i|\vec U_\bbi, \vec S, \vec Y) < \delta_N} \label{eq:lusy} \\
        \mathcal H_{U|S,Y} &= \set{i \in \bbN: Z(U_i|\vec U_\bbi, \vec S, \vec Y) > 1 - \delta_N} \label{eq:husy}
    \end{align}
    polarize with $\delta_N \triangleq 2^{-N^\beta}$ for any $\beta<\frac12$, yielding
    \begin{align}
        \lim_{N\to\infty} \frac1N\abs{\mathcal L_{U|S} \cup \mathcal H_{U|S,Y}} &= 1 - \MI(X; Y | S) \label{eq:conditional_polarization1} \\
        \lim_{N\to\infty} \frac1N\abs{\mathcal H_{U|S} \cap \mathcal L_{U|S,Y}} &= \MI(X; Y | S) \textnormal{.} \label{eq:conditional_polarization2}
    \end{align}
\end{lem}
\begin{proof}
    Using \cite[Eqs.~(38), (39)]{Honda13}, the sets \eqref{eq:lus} to \eqref{eq:husy} polarize analogously to \eqref{eq:asym_polarization1}, \eqref{eq:asym_polarization2}, i.e.,
    \begin{align}
        1 - \lim_{N\to\infty} \frac{\abs{\mathcal L_{U|S}}}{N}   &= \lim_{N\to\infty} \frac{\abs{\mathcal H_{U|S}}}{N}   &=& \He(X|S) \label{eq:s_polarized} \\ %
        1 - \lim_{N\to\infty} \frac{\abs{\mathcal L_{U|S,Y}}}{N} &= \lim_{N\to\infty} \frac{\abs{\mathcal H_{U|S,Y}}}{N} \!\!\! &=& \He(X|S,Y) \textnormal{.} \label{eq:sy_polarized} %
    \end{align}
    We next show \eqref{eq:conditional_polarization1}.
    Basic set theory gives
    \begin{equation}
        \abs{\mathcal L_{U|S} \cup \mathcal H_{U|S,Y}} = \abs{\mathcal L_{U|S}} + \abs{\mathcal H_{U|S,Y}} - \abs{\mathcal L_{U|S} \cap \mathcal H_{U|S,Y}} \textnormal{.} \label{eq:set_size}
    \end{equation}
    By \eqref{eq:bhattacharyya_conditional}, we have $\mathcal L_{U|S} \subseteq \mathcal L_{U|S,Y}$ and thus $\mathcal L_{U|S} \cap \mathcal H_{U|S,Y} = \emptyset$.
    Inserting the first term of \eqref{eq:s_polarized} and the second term of \eqref{eq:sy_polarized} into \eqref{eq:set_size} gives \eqref{eq:conditional_polarization1}.
    To prove \eqref{eq:conditional_polarization2}, observe that
    \begin{align}
        \lim_{N\to\infty} \frac{\abs{\mathcal H_{U|S} \cap \mathcal L_{U|S,Y}}}{N} &= 1 - \lim_{N\to\infty} \frac{\abs{(\mathcal H_{U|S} \cap \mathcal L_{U|S,Y})^C}}{N} \\
                                                                                &= 1 - \lim_{N\to\infty} \frac{\abs{\mathcal H_{U|S}^C \cup \mathcal L_{U|S,Y}^C}}{N} \\
                                                                                &= 1 - \lim_{N\to\infty} \frac{\abs{\mathcal L_{U|S} \cup \mathcal H_{U|S,Y}}}{N} \label{eq:set_complements}
    \end{align}
    where the last equality follows from \eqref{eq:s_polarized} and \eqref{eq:sy_polarized}.
    Combining \eqref{eq:conditional_polarization1} and \eqref{eq:set_complements} yields \eqref{eq:conditional_polarization2}.
\end{proof}

\begin{figure}[tb]
    \centerline{\scalebox{0.56}{\includegraphics{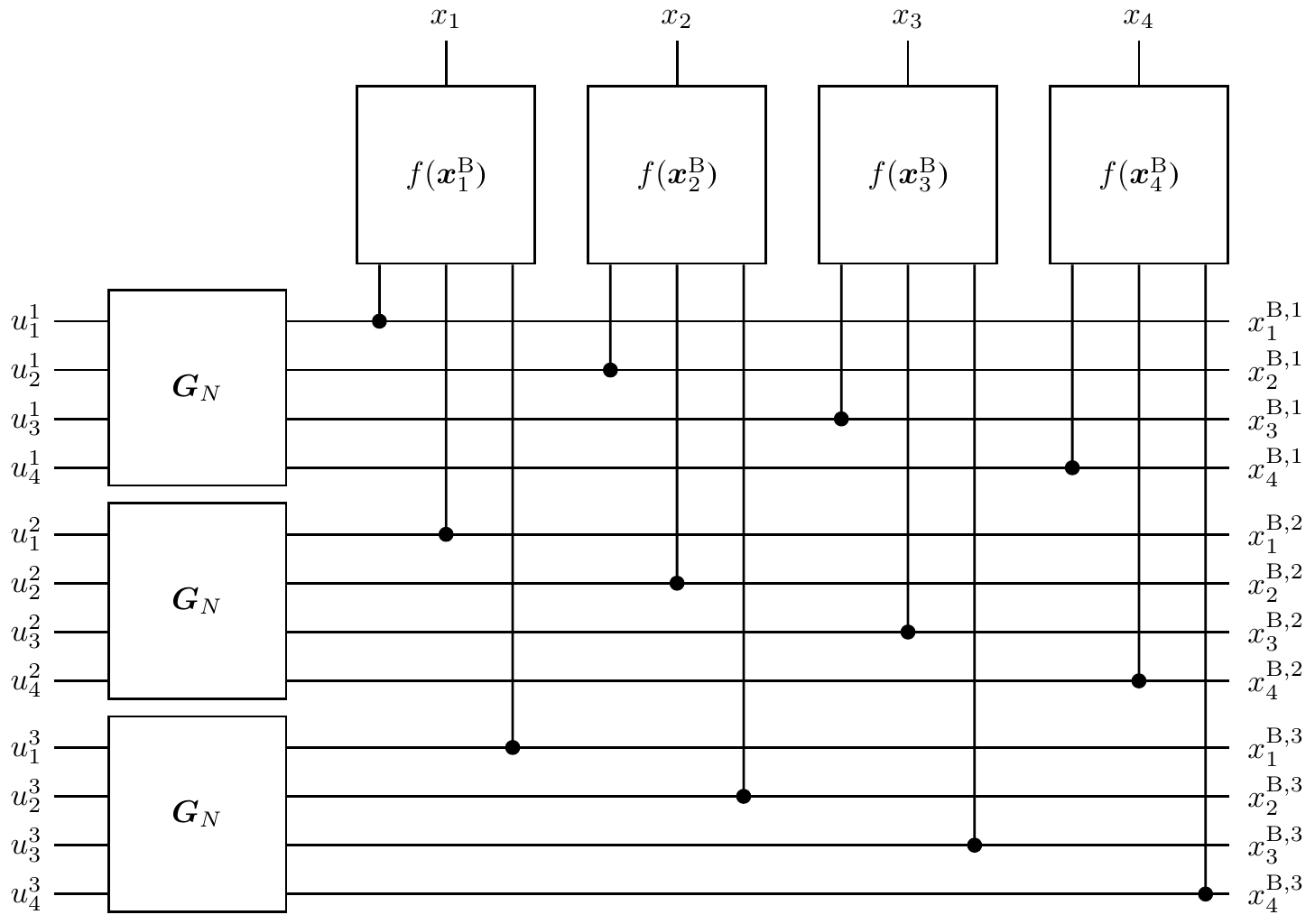}}}
    \caption{Multilevel polar-coded modulation with $M=2^m=8$ and $N=4$.}
    \label{fig:pcm}
\end{figure}

\section{Multilevel Polar Codes for DMCs}
\label{sec:dimc}
Consider the multilevel code construction in \cite{Seidl13} depicted in Fig.~\ref{fig:pcm}.
For a channel with input alphabet $\mathcal X$ of cardinality $M = \abs{\mathcal X} = 2^m$, each symbol $x$ is labelled with $m$ bits, \ie, $x = f(x^{\mathrm B, 1}x^{\mathrm B, 2}\dotsi x^{\mathrm B, m})$ where $f(\cdot)$ is invertible.
A codeword has a length of $N$ symbols or $mN$ bits.
Each bitlevel $\vec x^{\mathrm B, \ell}$, $\ell\in\bbm$, is encoded with a separate polar transform as $\vec u^\ell = \vec x^{\mathrm B, \ell} \matr G_N$. %
Using Lem.~\ref{lem:polarize_conditional_mi}, we can prove the polarization of such a multilevel polar code.

\begin{thm}
    \label{thm:multilevel_polarization}
    Let $W\colon X \to Y$ be a \gls{dmc} with joint distribution $X,Y \sim P_{Y|X}P_X$ and $\abs{\mathcal X} = 2^m$.
    Let $X^{\mathrm B, \ell}$, $\ell\in\bbm$, be the $\ell$-th bit of the binary representation $\vec X^{\mathrm B}$ of symbol $X = f(\vec X^{\mathrm B})$, and let $\vec U^\ell = \vec X^{\mathrm B, \ell} \matr G_N$.
    Then, the sets
    \begin{align}
        \mathcal L'_{U}   &= \set{(\ell,i): Z(U^\ell_i|\vec V^\ell_i) < \delta_N}             \\
        \mathcal H'_{U}   &= \set{(\ell,i): Z(U^\ell_i|\vec V^\ell_i) > 1 - \delta_N}         \\
        \mathcal L'_{U|Y} &= \set{(\ell,i): Z(U^\ell_i|\vec V^\ell_i, \vec Y) < \delta_N}     \\
        \mathcal H'_{U|Y} &= \set{(\ell,i): Z(U^\ell_i|\vec V^\ell_i, \vec Y) > 1 - \delta_N}
    \end{align}
    with $\ell\in\bbm$, $i\in\bbN$, and $\vec V^\ell_i = (\vec U^\ell_\bbi, \vec X^{\mathrm B, \bbl})$ polarize, \ie, we have
    \begin{align}
        \lim_{N\to\infty} \frac1N\abs{\mathcal L'_{U} \cup \mathcal H'_{U|Y}} &= 1 - \MI(X; Y) \label{eq:mlhy_cmi} \\
        \lim_{N\to\infty} \frac1N\abs{\mathcal I}
        &= \MI(X; Y) \label{eq:mlhy_mi}
    \end{align}
    where $\mathcal I = \mathcal H'_U \cap \mathcal L'_{U|Y}$.
\end{thm}
\begin{proof}
    We begin by showing \eqref{eq:mlhy_mi}.
    For each bitlevel $\ell\in\bbm$, consider the sets
    \begin{align}
        \mathcal H'_{U^\ell}   &= \set{(\ell, i): Z(U^\ell_i|\vec V^\ell_i) > 1 - \delta_N} \\
        \mathcal L'_{U^\ell|Y} &= \set{(\ell, i): Z(U^\ell_i|\vec V^\ell_i, \vec Y) < \delta_N} .
    \end{align}
    By Lem.~\ref{lem:polarize_conditional_mi}, we have
    \begin{equation}
        \lim_{N\to\infty} \frac1N\abs{\mathcal H'_{U^\ell} \cap \mathcal L'_{U^\ell|Y}} = \MI(X^{\mathrm B, \ell}; Y | \vec X^{\mathrm B, \bbl}) \textnormal{.}
    \end{equation}
    Note that $\mathcal H'_{U^\ell} \cap \mathcal H'_{U^k} = \mathcal L'_{U^\ell|Y} \cap \mathcal L'_{U^k|Y} = \mathcal H'_{U^\ell} \cap \mathcal L'_{U^k|Y} = \emptyset$ by definition for $\ell\neq k$  and
    \begin{equation}
        \mathcal H'_U = \bigcup_{\ell\in\bbm} \mathcal H'_{U^\ell}, \quad \mathcal L'_{U|Y} = \bigcup_{\ell\in\bbm} \mathcal L'_{U^\ell|Y} \textnormal{.}
    \end{equation}
    Additionally, since $f(\cdot)$ is invertible we have
    \begin{align}
        \lim_{N\to\infty}\frac1N\abs{\mathcal I} &= \lim_{N\to\infty} \frac1N\sum_{\ell\in\bbm} \abs{\mathcal H'_{U^\ell} \cap \mathcal L'_{U^\ell|Y}} \\
        &= \sum_{\ell\in\bbm} \MI(X^{\mathrm B, \ell}; Y | \vec X^{\mathrm B, \bbl}) \\
        &= \MI(X; Y)
    \end{align}
    which proves \eqref{eq:mlhy_mi}.
    The proof of \eqref{eq:mlhy_cmi} follows analogously to the proof of \eqref{eq:conditional_polarization2}.
\end{proof}

\subsection{Encoding and Decoding}
\label{sec:enc_dec}

The encoding is similar to \cite{Honda13}.
Define the message set $\mathcal M\subseteq\mathcal I$ as the set of bit positions $(\ell, i)$ populated by data bits.
The remaining bits $u^\ell_i$, $(\ell, i)\in\mathcal M^C$, are chosen successively and randomly with probability $P_{U^\ell_i|\vec V^\ell_i}(\cdot|\vec v^\ell_i)$, where $\vec v^\ell_i$ again includes the bits decided before $(\ell,i)$.
To compute $P_{U^\ell_i|\vec V^\ell_i}(\cdot|\vec v^\ell_i)$,  we factor $P_X(x_i)$ as
\begin{equation}
    P_X(x_i) = \prod_{\ell\in\bbm} P_{X^{\mathrm B,\ell}|\vec X^{\mathrm B,\bbl}}(x^{\mathrm B, \ell}_i | \vec x^{\mathrm B, \bbl}_i) \textnormal{.}
\end{equation}
For each bitlevel $\ell$, \gls{msd} computes $P_{X^{\mathrm B,\ell}|\vec X^{\mathrm B,\bbl}}(\cdot | \vec x^{\mathrm B, \bbl}_i)$, $i\in\bbN$ \cite{Imai77}, and provides these values to a \gls{sc} decoder that computes $P_{U^\ell_i|\vec V^\ell_i}(\cdot|\vec v^\ell_i)$ and decides on $u^\ell_i$.

The decoder uses the same \gls{msd} structure with \gls{sc} decoding.
The bits $\hat u^\ell_i$, $(\ell, i)\in\mathcal M$, are estimated as $\hat u^\ell_i = \argmax_{u} P_{U^\ell_i|\vec V^\ell_i,\vec Y}(u|\vec v^\ell_i, \vec y)$ assuming perfect knowledge of the previous bits $\vec v^\ell_i$.
The non-message bits are decided from $P_{U^\ell_i|\vec V^\ell_i}(\cdot|\vec v^\ell_i)$ requiring randomness that is shared by the transmitter and receiver.
The decoding error probability $\Pr(\set{\vec{\hat U} \neq \vec U})$ is averaged over this randomness.

\begin{thm}
    \label{thm:error_exp}
    Let $W\colon X \to Y$ and define $\mathcal I$ as in Thm.~\ref{thm:multilevel_polarization}.
    Let $\mathcal M \subseteq \mathcal I$, and consider encoding and decoding as described above.
    Then the average decoding error probability is $\Pr(\set{\vec{\hat U} \neq \vec U}) = \mathcal O(\nbp)$ for any $0<\beta'<\beta<1/2$ by choosing the polarization sets with $\delta_N = 2^{-N^\beta}$.
\end{thm}
\begin{proof}
    Consider
    \begin{equation}
        \Pr(\set{\vec{\hat U} \neq \vec U}) = 1 - \prod_{\ell\in\bbm} \left(1 - \Pr\mleft(\mathcal E^\ell \middle| \mathcal C^\bbl \mright)\right)
    \end{equation}
    where $\mathcal E^\ell = \set{\vec{\hat U}^\ell \neq \vec U^\ell}$ and $\mathcal C^\bbl = \set{\vec{\hat U}^\bbl = \vec U^\bbl}$. %
    Let the equivalent channel for the $\ell$-th bitlevel be the channel that has bit $\ell$ as input and bits $\bbl$ as side-information available at transmitter and receiver.
    By \cite[Thm.~4.3.9]{Liu16}, \cite[Thm.~3]{Honda13}, the \gls{hy} code over this equivalent channel for bitlevel $\ell$ has an average decoding error probability $\Pr\mleft(\mathcal E^\ell \middle| \mathcal C^\bbl \mright) = \mathcal O(\nbp)$ with $\beta'<\beta<1/2$ and uniformly chosen messages.
    
    Thus, for each bitlevel $\ell$ there is a positive constant $c_\ell$ and a block length $N_\ell$ so that $\Pr\mleft(\mathcal E^\ell \middle| \mathcal C^\bbl \mright) \le c_\ell\nbp$ for all $N > N_\ell$.
    By choosing $c = \max_{\ell\in\bbm} c_\ell$ and $N^* = \max_{\ell\in\bbm} N_\ell$, we can bound the error probability for any $\ell\in\bbm$ by
    \begin{equation}
        \Pr\mleft(\mathcal E^\ell \middle| \mathcal C^\bbl \mright) \le c\nbp
    \end{equation}
    for all $N > N^*$.
    The average decoding error probability under \gls{msd} can thus be bounded as
    \begingroup
    \allowdisplaybreaks
    \begin{align}
        \Pr(\set{\vec{\hat U} \neq \vec U}) &= 1 - \prod_{\ell\in\bbm} \left(1 - \Pr\mleft(\mathcal E^\ell \middle| \mathcal C^\bbl\mright)\right) \\
        &\le 1 - \left( 1 - c\,\nbp \right)^m \\
        &\le mc\,\nbp
    \end{align}
    \endgroup
    where the final step follows by Bernoulli's inequality.
\end{proof}

We remark that one can extend Thms.~\ref{thm:multilevel_polarization} and \ref{thm:error_exp} to discrete-input, continuous-output channels along the lines of \cite[Part~IV, Appendix~7]{Shannon48}.

\section{Short Blocklength Codes}
\label{sec:shaping}

A pragmatic approach is to choose the non-data bits with a deterministic rule \cite{Chou15}, \cite{Mondelli18} where the bits with large $\Huuly$ are fixed to $0$ and the bits with small $\Huul$ are decided as $\argmax_u P_{U_i|\vec V^\ell_i}(u|\vec{\hat v}^\ell_i)$ (``\gls{dm} bits'').
The remaining bits are data bits.
The decoder estimates the non-frozen bits via $\argmax_u P_{U_i|\vec V^\ell_i,\vec Y}(u|\vec{\hat v}^\ell_i, \vec y)$.
The \gls{dm} bits with $\Huul\approx0$ also have $\Huuly\approx0$ and are thus reliably estimated.
We call the resulting scheme \gls{mlhy} coding. 
The entropies used for code construction can be computed with, e.g., \gls{mc} integration or density evolution with Gaussian approximation \cite{Bocherer17}.
Similar to \cite{Seidl13}, we jointly compute the bitchannel entropies over all bitlevels.

\subsection{Distribution Matching}

Consider first a code that performs only \gls{dm}, \ie, there are no frozen bits. To evaluate the performance, we consider the rate loss \cite[Sec. V-B]{Bocherer15}, \cite{Wiegart19}, \cite{Bohnke20}
\begin{equation}
    \Delta_{\mathrm R} = \He(\widehat P_X) - R = \He(\widehat P_X) - \frac{\abs{\mathcal U}}{N}
\end{equation}
where $\widehat P_X$ is the empirical distribution of $X$, and $\mathcal U$ indexes the bitchannels with uniformly-distributed data bits.
Typically, $\mathcal U$ consists of the bitchannels with $\Huul\approx1$.

\begin{figure}[tb]
  \centerline{\includegraphics{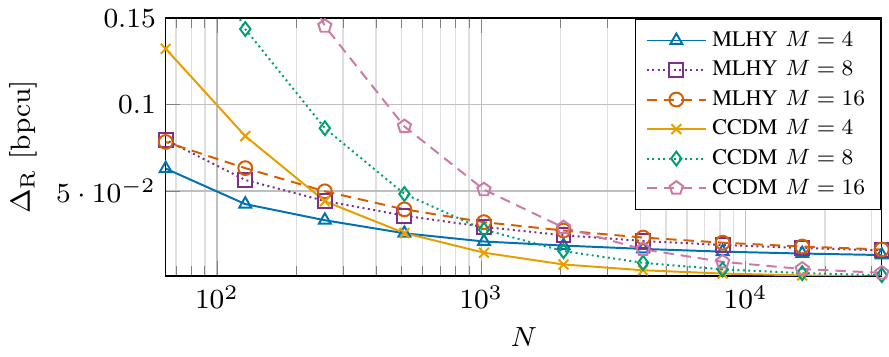}}
  \caption{Rate loss for MLHY with $L=32$ and \gls{ccdm} with different $M$.}
  \label{fig:mlhy_rate_loss}
\end{figure}

Fig.~\ref{fig:mlhy_rate_loss} shows the rate loss for \gls{ccdm} \cite{Schulte16} and for \gls{mlhy} \gls{dm} with \gls{scl} encoding \cite{Tal15} with list size $L=32$ instead of randomized encoding.
The target distributions are the maximum entropy distributions for the rates in \gls{bpcu}:
\begin{itemize}
    \item $R=\SI{1.625}{\bpcu}$ for $M=4$;
    \item $R=\SI{2.375}{\bpcu}$ for $M=8$;
    \item $R=\SI{3.250}{\bpcu}$ for $M=16$.
\end{itemize}
The \gls{mlhy} code is constructed by using the $RN$ bitchannels with the largest $\Huul$ for data.
The quantized distribution and the rate for \gls{ccdm} are determined by \cite[Algorithm~2]{Bocherer16} and \cite[Eq.~(37)]{Bocherer15}, respectively.

Observe that \gls{ccdm} is better than \gls{mlhy} \gls{dm} for large block lengths \cite{Schulte17}, \cite{Wiegart19}. This is expected since the polar code has a rigid structure. However, \gls{mlhy} \gls{dm} has a smaller rate loss than \gls{ccdm} for practically-relevant block lengths up to $N=1024$.  We observe that the rate loss of \gls{ccdm} increases with $M$ for short and moderate block lengths whereas the rate loss of \gls{mlhy} degrades only slightly. \gls{mlhy} \gls{dm} thus has superior performance for short block lengths and offers the flexibility to design joint \gls{dm} and \gls{fec} schemes.

\subsection{End-to-end Frame Error Rates}

We compare \gls{mlhy} coding with the \gls{pcpas} scheme proposed in \cite{Prinz17}.
\Gls{pcpas} uses the \gls{pas} architecture \cite{Bocherer15} with a systematic multilevel polar code as \gls{fec} and \gls{ccdm} \cite{Schulte16} for \gls{dm}.

Consider bipolar \gls{ask} and unipolar \gls{pam}. The input alphabets of cardinality $M=8$ and $M=4$, respectively, are
\begin{align*}
    \mathcal{X}_\mathrm{ASK} = \{\pm 7, \pm 5, \pm 3, \pm 1\}, \quad
    \mathcal{X}_\mathrm{PAM} = \{0,1,2,3\}.
\end{align*}
For both cases, we choose $P_X(x) \propto \exp(-\nu\abs{x}^2)$ so that $\nu$ minimizes the \gls{fer}.

The transmitter and receiver use \gls{scl} decoding with list size $L=32$ and an optional outer \gls{crc} code.
The code is designed for a specific rate $R$ and block length $N$.
There are three relevant design parameters.
The first is the \gls{dsnr} that determines the noise variance for computing $\Huuly$.
Second, we introduce a design parameter $\kappa$ for code optimization and choose the rate-optimal $P_X$ at $\kappa\cdot\mathrm{dSNR}$ as our target distribution based on which we also compute $\Huul$.
This parameter can improve the finite length performance because the optimal channel input distribution might deviate from the asymptotically optimal one.
Finally, we optimize over the number $N_\mathrm{DM}$ of \gls{dm} positions.
The code is constructed by choosing the $N_\mathrm{DM}$ positions with lowest $\Huul$ for \gls{dm} and the $N(1-R) - N_\mathrm{DM}$ positions with highest $\Huuly$ for \gls{fec}.
The remaining positions are used for data.
We use set-partitioning labelling \cite{Wachsmann99} for the channel input symbols.

The scheme from~\cite{Prinz17} must be modified to transmit PAM symbols with polar-coded \gls{pas}.
First, \gls{pas} requires symmetric distributions, as described in the introduction.
The one-sided sampled Gaussian distribution that we use for PAM does not fulfill this requirement.
Instead, we approximate a one-sided sampled Gaussian distribution $\tilde P_X$ for $M$-PAM by assigning $M/2$ different probability masses to pairs of points as described in~\cite{He19}, \cite{Kim20}, i.e., $\tilde P_X(2i) = \tilde P_X(2i+1)\; \forall i=0,1,\ldots,M/2-1$.
The input distribution is thus suboptimal.
Second, polar-coded \gls{pas} uses a set-partitioning labelling.
For ASK modulation as in~\cite{Prinz17}, the last bitlevel carries the sign of the constellation.
For PAM modulation, the first bitlevel refers to the bit that maps the transmitted signal to either the one or the other point of a pair.
This facilitates systematic encoding and we can omit the labelling transformation described in~\cite{Prinz17}.

\begin{figure}[tb]
    \centerline{\includegraphics{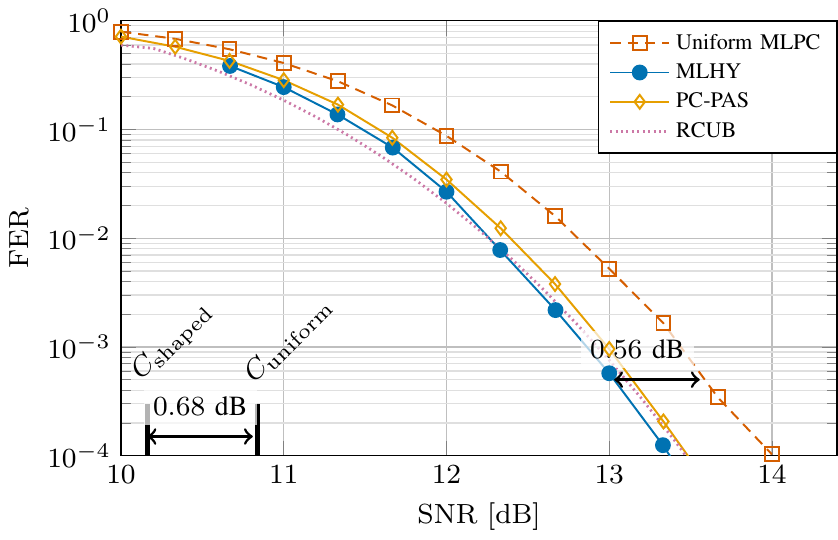}}
    \caption{Performance of MLHY coding compared to PC-PAS \cite{Prinz17}, uniform MLPC, and the RCUB, with an $8$-ASK constellation, $N=64$, $L=32$, at a $R = \SI{1.75}{\bpcu}$, shaped $\mathrm{dSNR} = \SI{13}{\dB}$, $\kappa = \SI{-1}{\dB}$, $N_\mathrm{DM} = 23$, and uniform $\mathrm{dSNR} = \SI{13}{\dB}$. The MLHY and MLPC codes use an outer CRC-7 and PC-PAS a CRC-4 together with a type check. For PC-PAS, $\mathrm{dSNR}=\SI{8}{dB}$ and $\kappa=\SI{-0.6}{dB}$.} 
    \label{fig:mlhy_fer_ask8}
\end{figure}

Fig.~\ref{fig:mlhy_fer_ask8} shows the \gls{fer} for an $8$-\gls{ask} constellation and $N=64$.
We also show the \gls{rcub} \cite{Polyanskiy10} computed for the distribution realized by the \gls{mlhy} encoder, and the \gls{fer} for uniform \gls{mlpc} \cite{Seidl13}.
The codes and bounds are designed for $R=\SI{1.75}{\bpcu}$.
The bold black lines at \SI{10.16}{\dB} and \SI{10.84}{\dB} show the constellation-constrained capacities for shaped and uniform transmission, respectively.

The error curve slopes for \gls{mlhy} coding and uniform \gls{mlpc} are similar, resulting in an almost constant shaping gain in the waterfall region.
Both \gls{mlhy} coding and \gls{pcpas} perform close to the theoretical shaping gain of \SI{0.68}{\dB} and to the \gls{rcub}. The \gls{mlhy} scheme thus performs on par with \gls{pcpas}, even without a dedicated code optimization beyond a random search over the design paramters.

We describe potential improvements.
Because \gls{ccdm} codewords are all of the same type, \gls{pcpas} permits an additional list pruning step \cite{Prinz17} so that the length of the outer \gls{crc} code can be reduced.
The performance of \gls{mlhy} coding may be improved by further adjusting the design parameters, optimizing the bitchannel selection process, optimizing the \gls{crc} polynomial and length, and checking candidate codewords against \gls{dm} constraints at the decoder.

\begin{figure}[tb]
    \centerline{\includegraphics{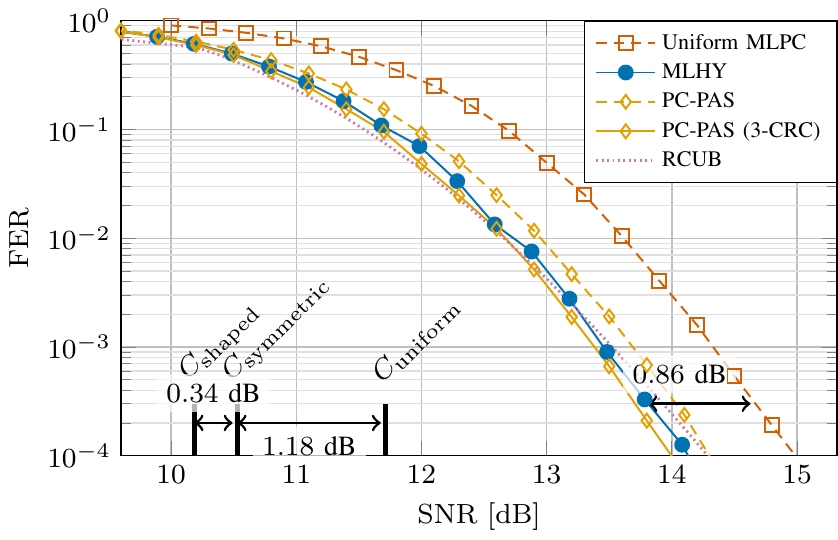}}
    \caption{Performance of MLHY coding compared to PC-PAS, uniform MLPC, and the RCUB, with an $4$-PAM constellation, $N=64$, $L=32$, at a $R = \SI{1.25}{\bpcu}$, shaped $\mathrm{dSNR} = \SI{18.1}{\dB}$, $\kappa= \SI{-0.9}{\dB}$, $N_\mathrm{DM} = 24$, and uniform $\mathrm{dSNR} = \SI{19.25}{\dB}$. The MLHY and MLPC codes do not use an outer CRC. For PC-PAS we depict curves with and without outer CRC. Further, for PC-PAS, $\mathrm{dSNR} = \SI{14.5}{dB}$ and $\kappa = \SI{-3.9}{dB}$.}
    \label{fig:mlhy_fer_im4}
\end{figure}

Fig.~\ref{fig:mlhy_fer_im4} depicts the \glspl{fer} for a $4$-PAM constellation with $N=64$ and $R=\SI{1.25}{\bpcu}$. We show shaped \gls{mlhy} coding and \gls{pcpas}, uniform \gls{mlpc}, and the \gls{rcub} for \gls{pam} over the \gls{awgn} channel.
The additional shaping gain of using the rate-optimal, asymmetric $P_X$ over a symmetric distribution $\tilde P_X$ is approximately \SI{0.34}{\dB}.
Without \gls{crc}, the \gls{mlhy} curve exhibits the predicted shaping gain and outperforms \gls{pcpas} without \gls{crc}. It further lies on top of the \gls{rcub}.
With list pruning by \gls{crc} and type checking, \gls{pcpas} gains approximately $\SI{0.2}{dB}$. %
First results using an additional outer \gls{crc} code in the \gls{mlhy} scheme did not provide a noticeable coding gain.
We therefore did not include the \gls{crc} curves in this case. We expect to recover the full shaping gain by further optimizing the polar and \gls{crc} codes.

Recall that \gls{mlhy} coding uses the same binary polar multistage decoder at the transmitter and the receiver.
The implementation complexity is thus reduced as compared to \gls{pas}.
Furthermore, the use of \gls{ccdm} as an outer code causes the end-to-end \gls{ber} of \gls{pas} to typically be much higher than for \gls{mlhy} coding for the same \gls{fer}.

\section{Conclusion}

We showed that multilevel polar coded modulation with binary polar codes and Honda-Yamamoto probabilistic shaping can achieve the capacity of \glspl{dmc} with input alphabets of cardinality a power of two. 
The performance is on-par with state-of-the-art \gls{pcpas} for short and moderate block lengths.
Future research may further optimize the code for these block lengths, and investigate how the constraints induced by the deterministic \gls{dm} process can be used to aid decoding.

\section*{Acknowledgment}
The authors wish to thank Prof. Gerhard Kramer for suggestions.
This work was supported in part by the German Federal Ministry of Education and Research (BMBF) under the Grant 6G-life, and by the German Research Foundation (DFG) under Grant KR 3517/9-1.

\clearpage
\bibliographystyle{IEEEtran}
\bibliography{IEEEabrv,confs-jrnls,manual}

\end{document}